\documentclass[a4paper,envcountsame]{llncs}

\usepackage{ifthen}
\usepackage{color}
\usepackage{amsmath,amsfonts,amssymb}
\usepackage{mathtools}
\usepackage{txfonts}

\usepackage{tikz}
\usetikzlibrary{arrows,shapes,backgrounds}

\tikzstyle{max}=[shape=rectangle,draw,inner sep=0pt,minimum size=6mm,thick]
\tikzstyle{min}=[shape=diamond,draw,inner sep=0pt,minimum size=6mm,thick]
\tikzstyle{ran}=[shape=circle,draw,inner sep=0pt,minimum size=6mm,thick]
\tikzstyle{tran}=[thick,draw,->,>=stealth]

\newcommand{\Rset}{\mathbb{R}}

\newcommand{\E}{\mathbb{E}}
\newcommand{\G}{\mathcal{G}}
\newcommand{\M}{\mathcal{M}}

\newcommand{\calF}{\mathcal{F}}
\newcommand{\calP}{\mathcal{P}}
\newcommand{\Nset}{\mathbb{N}}
\newcommand{\Qset}{\mathbb{Q}}
\newcommand{\eps}{\varepsilon}
\newcommand{\Rni}{\Rset^{\geq 0}_\infty}
\newcommand{\Ord}{\mathit{Ord}}

\newcommand{\vleq}{\sqsubseteq}

\newcommand{\genTran}[2]{%
    {}\mathchoice%
    {\stackrel{#1}{#2}}
    {\mathop {\smash{#2}}\limits^{\vrule width 0pt height 0pt depth 4pt\smash{#1}}}
    {\stackrel{#1}{#2}}
    {\stackrel{#1}{#2}}
{}}
\newcommand{\tran}[1]{\genTran{#1}{\rightarrow}}

\newcommand{\Exp}{\mathbb{E}}

\newcommand{\Acc}{\mathit{Acc}}
\newcommand{\Reach}{\mathit{Reach}}
\newcommand{\Payoff}{\mathit{Payoff}}

\newcommand{\Prob}{\mathit{Prob}}
\newcommand{\fpath}{\mathit{Fpath}}
\newcommand{\run}{\mathit{Run}}

\newcommand{\val}{\mathrm{Val}}

\newcommand{\last}[1]{\mathrm{last(#1)}}

\newcommand{\MD}{\mathrm{MD}}
\newcommand{\MR}{\mathrm{MR}}
\newcommand{\HD}{\mathrm{HD}}
\newcommand{\HR}{\mathrm{HR}}

\newcommand{\ifApp}[2]%
{\ifthenelse{\isundefined{\showappendix}}{#2}{#1}}

\spnewtheorem{fact}[theorem]{Fact}{\bfseries}{\itshape}

\newcommand{\theoremlike}[2]{\par\medskip\penalty-250%
{{\bfseries\noindent
#2 \ref{#1}.}}\it}

\newcommand{\thmhelperpre}[2]{\theoremlike{#1}{#2}}
\newcommand{\thmhelperpost}{\par\medskip}

\newenvironment{reflemma}[1]{\thmhelperpre{#1}{Lemma}}{\thmhelperpost}

\title{Determinacy in Stochastic Games with Unbounded Payoff Functions}
\titlerunning{Games with Unbounded Payoffs}

\author{%
Tom\'{a}\v{s} Br\'{a}zdil\thanks{The authors are supported by 
the Czech Science Foundation, grant No.~P202/12/G061.}
\and Anton\'{\i}n Ku\v{c}era$^\star$ \and
Petr Novotn\'{y}$^\star$}

\authorrunning{
Br\'{a}zdil \and
Ku\v{c}era \and
Novotn\'{y}
}

\institute{
Faculty of Informatics, Masaryk University \\
\email{\{xbrazdil,kucera,xnovot18\}@fi.muni.cz}
}
\begin{document}

\maketitle

\begin{abstract}
  We consider infinite-state turn-based stochastic games of two
  players, $\Box$~and~$\Diamond$, who aim at maximizing and minimizing
  the expected total reward accumulated along a run,
  respectively. Since the total accumulated reward is unbounded, the
  determinacy of such games cannot be deduced directly from Martin's
  determinacy result for Blackwell games.
  Nevertheless, we show that these games \emph{are} determined both
  for unrestricted (i.e., history-dependent and randomized) strategies
  and deterministic strategies, and the equilibrium value is the
  same. Further, we show that these games are generally \emph{not}
  determined for memoryless strategies. Then, we consider a subclass
  of \emph{$\Diamond$-finitely-branching} games and show that they are
  determined for all of the considered strategy types, where the
  equilibrium value is always the same.
  We also examine the existence and type of ($\varepsilon$-)optimal
  strategies for both players.
\end{abstract}

\section{Introduction}
\label{sec-intro}

Turn-based stochastic games of two players are  
a standard model of discrete systems that exhibit both non-deterministic and
randomized choice. One player (called $\Box$ or Max in this paper) 
corresponds to the controller who wishes to achieve/maximize some 
desirable property of the system, and the other player (called $\Diamond$ 
or Min) models the environment which aims at spoiling the 
property. Randomized choice is used to model events such as system 
failures, bit-flips, or coin-tossing in randomized algorithms.

Technically, a turn-based stochastic game (SG) is defined
as a directed graph where every vertex is either stochastic or
belongs to one of the two players. Further, there is a fixed probability
distribution over the outgoing transitions of every stochastic 
vertex. A \emph{play} of the game is initiated by putting a token
on some vertex. Then, the token is moved from vertex to vertex
by the players or randomly. A \emph{strategy} specifies how a player
should play. In general, a strategy may depend on the sequence of vertices
visited so far (we say that the strategy is \emph{history-dependent (H)}),
and it may specify a probability distribution over the outgoing transitions
of the currently visited vertex rather than a single outgoing transtion
(we say that the strategy is \emph{randomized (R)}). Strategies that
do not depend on the history of a play are called \emph{memoryless (M)},
and strategies that do not randomize (i.e., select a single outgoing 
transition) are called \emph{determinisctic (D)}. Thus, we obtain the
MD, MR, HD, and HR strategy classes, where HR are unrestricted strategies
and MD are the most restricted memoryless deterministic strategies.  

A \emph{game objective} is usually specified by a \emph{payoff function}
which assigns some real value to every run (infinite path) in the
game graph. The aim of Player~$\Box$ is to \emph{maximize} the expected
payoff, while Player~$\Diamond$ aims at \emph{minimizing} it. 
It has been shown in \cite{MS:stochastic-games-determinacy} that
for \emph{bounded} and \emph{Borel} payoff functions,
Martin's determinacy result for Blackwell 
games~\cite{Martin:Blackwell-determinacy} implies that
\begin{equation}
\label{eq-determinacy}
   \adjustlimits\sup_{\sigma \in \HR_{\Box}} \inf_{\pi \in \HR_{\Diamond}} 
       \Exp^{\sigma,\pi}_v[\Payoff]
   \quad = \quad
   \adjustlimits\inf_{\pi \in \HR_{\Diamond}} \sup_{\sigma \in \HR_{\Box}} 
       \Exp^{\sigma,\pi}_v[\Payoff]
\end{equation}
where $\HR_{\Box}$ and $\HR_{\Diamond}$ are the classes of HR strategies for
Player~$\Box$ and Player~$\Diamond$, respectively.
Hence, every vertex $v$ has a \emph{HR value} $\val_{\HR}(v)$ specified 
by~\eqref{eq-determinacy}. A HR strategy is \emph{optimal} if it achieves 
the outcome $\val_{\HR}(v)$ or better against every strategy of the other 
player. In general,  optimal strategies are not guaranteed to exist,
but~\eqref{eq-determinacy} implies that both players have 
\emph{$\varepsilon$-optimal} HR strategies for every 
$\varepsilon > 0$ (see Section~\ref{sec-prelim} for precise definitions).

The determinacy results 
of~\cite{Martin:Blackwell-determinacy,MS:stochastic-games-determinacy} cannot
be applied to \emph{unbounded} payoff functions,
i.e., these results do not imply 
that \eqref{eq-determinacy} holds if $\Payoff$ is unbounded,
and they do not say anything about the existence of a value for 
restricted strategy classes such as MD or MR. 
In the context of performance analysis and controller synthesis, 
these questions rise naturally; in some cases, the players cannot 
randomize or remember the history of a play, and some of the
studied payoff functions are not bounded. In this paper, we study
these issues for the \emph{total accumulated reward} payoff function
and \emph{infinite-state} games.

The total accumulated reward payoff function, denoted by $\Acc$, 
is defined as follows. Assume that every vertex $v$ is assigned a fixed
non-negative reward $r(v)$. Then $\Acc$ assigns to every run the sum
of rewards all vertices visited along the run. Obviously, $\Acc$ is
unbounded in general, and may even take the $\infty$ value.  A
special case of total accumulated reward is \emph{termination time},
where all vertices  are assigned reward~$1$, except for terminal 
vertices that are assigned reward~$0$ (we
also assume that the only outgoing transition of every terminal
vertex~$t$ is a self-loop on~$t$). Then, $\Exp^{\sigma,\pi}_v[\Acc]$
corresponds to the expected termination time under the strategies
$\sigma,\pi$.  Another special (and perhaps simplest) case of total
accumulated reward is \emph{reachability}, where the target vertices
are assigned reward~$1$ and the other vertices have zero reward 
(here we assume that every target vertex has a single outgoing 
transition to a special state $s$ with zero reward, where 
$s \tran{} s$ is the only outgoing transition of $s$). Although the 
reachability payoff is bounded, some of our negative results about 
the total accumulated reward hold even for
reachability (see below).

The reason for considering infinite-state games is that many
recent works study various algorithmic problems for games over
classical automata-theoretic models, such as pushdown automata
\cite{EY:RMC-RMDP,EY:RMDP-efficient,EY:RSCG,EWY:RSG-Positive-Rewards,%
BBKO:BPA-games-reachability-IC,BBFK:BPA-games-reachability-IC}, lossy
channel systems \cite{BBS:lossy-games-reach,%
AHAMS:Stochastic-games-lossy}, one-counter automata
\cite{BBEKW:OC-MDP,BBE:OC-games,BBEK:OC-games-termination-approx}, 
or multicounter automata
\cite{FJLS:multi-energy-games,BJK:eVASS-games,BCKN:consumption-games,%
Kucera:multicounter-games,CHDHR:energy-mean-payoff,%
BFLMS:weighted-automata-inf-runs}, which are finitely
representable but the underlying game graph is infinite and sometimes
even infinitely-branching (see, e.g., 
\cite{BJK:eVASS-games,BCKN:consumption-games,Kucera:multicounter-games}).
Since the properties of finite-state games do \emph{not} carry over
to infinite-state games in general (see, e.g., \cite{Kucera:games-chapter}),  
the above issues need to be revisited and clarified explicitly, which
is the main goal of this paper.

\textbf{Our contribution:} We consider general
infinite-state games, which may contain vertices with infinitely
many outgoing transitions, and 
$\Diamond$-finitely-branching games, where every vertex of $V_\Diamond$ 
has finitely many outgoing transitions, with the total accumulated 
reward objective. For \emph{general} games, we show the following:
\begin{itemize}
\item Every vertex has both a HR and a HD value, and these values are 
  equal\footnote{For a given strategy type $T$ (such as MD or MR), 
  we say that a vertex $v$ has a \emph{$T$~value} if 
  $\textstyle\sup_{\sigma \in T_{\Box}} \inf_{\pi \in T_{\Diamond}} 
       \Exp^{\sigma,\pi}_v[\Payoff]
   \ = \  
   \textstyle\inf_{\pi \in T_{\Diamond}} \sup_{\sigma \in T_{\Box}} 
       \Exp^{\sigma,\pi}_v[\Payoff]$, where $T_{\Box}$ and $T_{\Diamond}$
   are the classes of all $T$ strategies for Player~$\Box$ and 
   Player~$\Diamond$, respectively.}. 
\item There is a vertex $v$ of a game $G$ with reachability objective
  such that $v$ has neither MD nor MR value. Further, the game
  $G$ has only one vertex (belonging to Player~$\Diamond$) with
  infinitely many outgoing transitions.
\end{itemize}
It follows from previous works (see, e.g., 
\cite{BBFK:BPA-games-reachability-IC,Kucera:games-chapter}) 
that optimal strategies in general games may not exist, and even if they
do exist, they may require infinite memory. Interestingly, we observe
that an optimal strategy for Player~$\Box$ (if it exists) may also 
require randomization in some cases.

For \emph{$\Diamond$-finitely-branching} games, we prove the following results:
\begin{itemize}
\item Every vertex has a HR, HD, MR, and MD~value, and all of these values
  are equal.
\item Player~$\Diamond$ has an optimal MD strategy in every vertex.
\end{itemize}
It follows from the previous works that Player~$\Box$ may not have
an optimal strategy and even if he has one, it may require infinite
memory. Let us note that in finite-state games, both players have
optimal MD strategies (see, e.g., \cite{FV:book}).

Our results are obtained by generalizing the arguments for reachability
objectives presented in \cite{BBFK:BPA-games-reachability-IC}, but there are
also some new observations based on original ideas and new 
counterexamples. In particular, this applies to the existence 
of a HD value and the non-existence of MD and MR values in general games.


\section{Preliminaries}
\label{sec-prelim}

In this paper, the sets of all positive integers, non-negative
integers, rational numbers, real numbers, and non-negative real
numbers are denoted by $\Nset$, $\Nset_0$, $\Qset$, $\Rset$, and
$\Rset^{\geq 0}$, respectively. We also use $\Rset^{\geq 0}_\infty$
to denote the set $\Rset^{\geq 0} \cup \{\infty\}$, where $\infty$
is treated according to the standard conventions. For all 
$c \in \Rset^{\geq 0}_\infty$ and $\varepsilon \in [0,\infty)$, we define
the \emph{lower} and \emph{upper} $\eps$-approximation of~$c$,
denoted by $c\ominus\eps$ and $c\oplus\eps$, respectively, as follows:
\[
\begin{array}{rcll}
  c \oplus \eps & = &  c + \eps \quad & 
      \mbox{for all $c \in \Rset^{\geq 0}_\infty$ 
            and $\varepsilon \in [0,\infty)$,}\\
  c \ominus\eps  & = & c - \eps  & 
      \mbox{for all $c \in \Rset^{\geq 0}$ 
            and $\varepsilon \in [0,\infty)$,}\\
  \infty \ominus\eps & = & 1/\eps &
      \mbox{for all $\varepsilon \in (0,\infty)$,}\\
  \infty \ominus 0 &\ =\ & \infty\, .
\end{array}
\]
Given a set $V$, the elements of $(\Rset^{\geq 0}_\infty)^V$ are written as vectors
$\vec{x},\vec{y},\ldots$, where $\vec{x}_v$ denotes the $v$-component
of $\vec{x}$ for every $v \in V$. The standard component-wise
ordering on $(\Rset^{\geq 0}_\infty)^V$ is denoted by $\sqsubseteq$.


For every finite or countably infinite set $M$, a binary relation 
${\to} \subseteq M \times M$ is \emph{total} if for every
$m \in M$ there is some $n \in M$ such that $m \to n$.
A \emph{finite path} in $\M = (M,{\to})$ is 
a finite sequence $w = m_0,\ldots,m_k$ such that 
$m_i \to m_{i+1}$ for every $i$, where $0 \leq i < k$. 
The \emph{length} of~$w$, i.e., the number of transitions performed 
along~$w$, is denoted by $|w|$.
A \emph{run} in  $\M$ is an infinite sequence 
$\omega = m_0,m_1,\ldots$ every finite prefix of which is 
a path. We also use $\omega(i)$ to denote the element
$m_i$ of $\omega$, and $\omega_i$ to denote the run $m_i,m_{i+1},\ldots$ 
Given $m,n \in M$, we say that $n$ is \emph{reachable} from $m$,
written $m \to^* n$, if there is a finite path 
from $m$ to $n$. The sets of all
finite paths and all runs in $\M$ are denoted by 
$\fpath(\M)$ and $\run(\M)$, respectively. For every finite path
$w$, we use $\run(\M,w)$ and $\fpath(\M,w)$ to denote the set of all runs and finite paths, respectively, prefixed by~$w$.
If $\M$ is clear from the context, we write just $\run$, $\run(w)$, $\fpath$ and $\fpath(w)$
instead of $\run(\M)$, $\run(\M,w)$, $\fpath(\M)$ and $\fpath(\M,w)$, respectively. 

Now we recall basic notions of probability theory. 
Let $A$ be a finite or countably infinite set. A 
\emph{probability distribution}
on $A$ is a function $f : A \rightarrow \Rset^{\geq 0}$ such that
\mbox{$\sum_{a \in A} f(a) = 1$}. A distribution $f$ is \emph{rational}
if $f(a) \in \Qset$ for every $a \in A$,
\emph{positive} if $f(a) > 0$ for every $a \in A$, \emph{Dirac}
if $f(a) = 1$ for some $a \in A$, and \emph{uniform} if $A$ is
finite and $f(a) = \frac{1}{|A|}$ for every $a \in A$.
A \emph{$\sigma$-field} over a set $X$ is a set $\calF \subseteq 2^X$
that includes $X$ and is closed under complement and countable union. 
A \emph{measurable space} is a pair $(X,\calF)$ where $X$ is a
set called \emph{sample space} and $\calF$ is a $\sigma$-field over $X$. 
A \emph{probability measure} over a measurable space
$(X,\calF)$ is a function $\calP : \calF \rightarrow \Rset^{\geq 0}$
such that, for each countable collection $\{X_i\}_{i\in I}$ of pairwise
disjoint elements of $\calF$, $\calP(\bigcup_{i\in I} X_i) = 
\sum_{i\in I} \calP(X_i)$, and moreover $\calP(X)=1$. A 
\emph{probability space} is a triple $(X,\calF,\calP)$ where 
$(X,\calF)$ is a measurable space and $\calP$
is a probability measure over $(X,\calF)$.

\begin{definition}
A \emph{stochastic game} is a tuple 
$G = (V,\tran{},(V_{\Box},V_\Diamond,V_{\bigcirc}),\Prob)$ where
$V$ is a finite or countably infinite set of \emph{vertices}, ${\tran{}}
\subseteq V \times V$ is a total \emph{transition relation},
$(V_{\Box},V_\Diamond,V_{\bigcirc})$ is a partition of $V$, and $\Prob$ is a
\emph{probability assignment} which to each $v \in V_{\bigcirc}$
assigns a positive probability distribution on the set of its outgoing
transitions. We say that $G$ is
\emph{$\Diamond$-finitely-branching} if for each $v \in V_{\Diamond}$ 
there are only finitely many $u \in V$ such that $v \tran{} u$. 
\end{definition}

\paragraph{\bfseries Strategies.}
A stochastic game $G$ is played by two players, $\Box$ and $\Diamond$,
who select the moves in the vertices of $V_{\Box}$ and $V_{\Diamond}$,
respectively. Let $\odot \in \{\Box,\Diamond\}$.
A \emph{strategy} for Player~$\odot$ in $G$ is a 
function which to each finite path in $G$ ending a vertex $v \in V_{\odot}$
assigns a probability distribution on the set of outgoing transitions of~$v$.
We say that a strategy 
$\tau$ is \emph{memoryless (M)} if $\tau(w)$ depends just on the last
vertex of $w$, and \emph{deterministic (D)} if it returns a Dirac
distribution for every argument. Strategies that are not
necessarily memoryless are called \emph{history-dependent (H)}, and
strategies that are not necessarily deterministic are called
\emph{randomized (R)}. Thus, we obtain the MD, MR, HD, and HR
\emph{strategy types}. The set of all strategies for Player~$\odot$
of type $T$ in a game $G$ is denoted by $T^G_{\odot}$, or just
by $T_{\odot}$ if $G$ is understood (for example, 
$\MR_{\Box}$ denotes the set of all MR strategies for Player~$\Box$).  

Every pair of strategies $(\sigma,\pi) \in \HR_\Box \times \HR_\Diamond$ 
and an initial vertex~$v$ determine a unique probability space
$(\run(v),\calF,\calP_v^{\sigma,\pi})$, where $\calF$ is the $\sigma$-field 
over $\run(v)$ generated by all $\run(w)$ such that 
$w$~starts with~$v$, and $\calP_v^{\sigma,\pi}$ is the unique
probability measure such that for every finite path 
$w = v_0,\ldots ,v_k$ initiated in $v$ we have that
$\calP_v^{\sigma,\pi}(\run(w)) = \Pi_{i=0}^{k-1} x_i$, where
$x_i$ is the probability of $v_{i} \tran{} v_{i+1}$
assigned either by  $\sigma(v_0,\ldots,v_{i})$, 
$\pi(v_0,\ldots,v_{i})$, or $\Prob(v_{i})$, depending 
on whether $v_{i}$ belongs to $V_{\Box}$, $V_\Diamond$, or 
$V_{\bigcirc}$, respectively (in the case when $k=0$, i.e., $w = v$, we
put $\calP_v^{\sigma,\pi}(\run(w)) = 1$). 

\paragraph{\bfseries Determinacy, optimal strategies.}

In this paper, we consider games with the \emph{total accumulated
reward} objective and \emph{reachability} objective, where the latter
is understood as a restricted form of the former (see below).

Let $r : V \rightarrow \Rset^{\geq 0}$ be a \emph{reward function},
and $\Acc : \run \rightarrow \Rset_{\infty}^{\geq 0}$ a function
which to every run~$\omega$ assigns the \emph{total accumulated
reward} $\Acc(\omega) = \sum_{i=0}^{\infty} r(\omega(i))$. 
%
Let $T$ be a strategy type. We say that a vertex $v \in V$
\emph{has a $T$-value} in $G$ if
\[
   \adjustlimits\sup_{\sigma \in T_{\Box}} \inf_{\pi \in T_{\Diamond}} 
       \Exp^{\sigma,\pi}_v[\Acc]
   \quad = \quad
   \adjustlimits\inf_{\pi \in T_{\Diamond}} \sup_{\sigma \in T_{\Box}} 
       \Exp^{\sigma,\pi}_v[\Acc]
\]
where $\Exp^{\sigma,\pi}_v[\Acc]$ denotes the expected value of $\Acc$ in
$(\run(v),\calF,\calP_v^{\sigma,\pi})$. If $v$ has a \mbox{$T$-value}, 
then $\val_T(v,r,G)$ (or just $\val_T(v)$ if $G$ and $r$ are clear
from the context) denotes the \emph{$T$-value of~$v$} defined by this equality.

Let $\G$ be a class of games. If every vertex of every $G \in \G$
has a $T$-value for every reward function, we say that $\G$ is 
\emph{$T$-determined}. Note that $\Acc$ is generally not bounded, 
and therefore we cannot directly apply the results of 
\cite{Martin:Blackwell-determinacy,MS:stochastic-games-determinacy}
to conclude that the class of all games is HR-determined. Further, 
these results do not say anything about 
determinacy for the other strategy types even for bounded objective 
functions.

If a given vertex $v$ has a $T$-value, we can define the notion
of $\varepsilon$-optimal $T$~strategy for both players.

\begin{definition}
\label{def-eps-opt}
  Let $v$ be a vertex which has a $T$-value, and let 
  $\varepsilon \geq 0$. We say that 
  \begin{itemize}
  \item  $\sigma \in T_{\Box}$ is \emph{$\varepsilon$-$T$-optimal}
    in $v$ if $\Exp^{\sigma,\pi}_v[\Acc] \geq \val_T(v) \ominus \varepsilon$ 
    for all $\pi \in  T_{\Diamond}$;
  \item $\pi \in T_{\Diamond}$ is \emph{$\varepsilon$-$T$-optimal}
    in $v$ if 
    $\Exp^{\sigma,\pi}_v[\Acc] \leq \val_T(v) \oplus \varepsilon$ for all
    $\sigma \in  T_{\Box}$.
  \end{itemize}
  A $0$-$T$-optimal strategy is called \emph{$T$-optimal}.
\end{definition}

In this paper we also consider \emph{reachability} objectives,
which can be seen as a restricted form of the total accumulated
reward objectives introduced above. A ``standard'' definition of
the reachability payoff function looks as follows: We fix
a set $R \subseteq V$ of \emph{target} vertices, and define
a function $\Reach : \run \rightarrow \{0,1\}$ which to every 
run assigns either $1$ or $0$ depending on whether or not
the run visits a target vertex. Note that 
$\Exp^{\sigma,\pi}_v[\Reach]$ is the \emph{probability} of visiting a 
target vertex in the corresponding play of~$G$. Obviously, if we
assign reward~$1$ to the target vertices and $0$ to the others,
and replace all outgoing transitions of target vertices with a single
transition leading to a fresh stochastic vertex $u$ with reward~$0$ and 
only one transition $u \tran{} u$, then $\Exp^{\sigma,\pi}_v[\Reach]$
in the original game is equal to $\Exp^{\sigma,\pi}_v[\Acc]$ in the modified
game. Further, if the original game was $\Diamond$-finitely-branching or finite,
then so is the modified game. Therefore, all ``positive'' results 
about the total accumulated reward objective
(e.g., determinacy, existence of \mbox{$T$-optimal} strategies, etc.) 
achieved in this paper
carry over to the reachability objective, and all ``negative''
results about reachability carry over to the total accumulated reward.


\section{Results}
\label{sec-results}

Our main results about the determinacy of general stochastic games 
with the total accumulated reward payoff function are summarized 
in the following theorem:

\begin{theorem}
  \label{thm:general-determinacy}
  Let $\G$ be the class of all games. Then 
  \begin{itemize}
  \item[a)] $\G$ is both HR-determined and HD-determined. Further, for
    every vertex $v$ of every $G \in \G$ and every reward function $r$
    we have that $\val_{\HR}(v) = \val_{\HD}(v)$.
  \item[b)] $\G$ is neither MD-determined nor MR-determined, and these results
    hold even for reachability objectives.
  \end{itemize}
\end{theorem}

\noindent 
An optimal strategy for Player~$\Box$ does not necessarily exist,
even if~$G$ is a game with a reachability payoff function
such that $V_\Diamond = \emptyset$ and every vertex of $V_\Box$ has
at most two outgoing transitions (see, 
e.g.,~\cite{BBFK:BPA-games-reachability-IC,Kucera:games-chapter}).
In fact, it suffices to consider the vertex~$v$ of Fig.~\ref{fig-determ}
where the depicted game is modified by replacing the vertex $u$ with
a stochastic vertex $u'$, where $u' \tran{} u'$ is the
only outgoing transition of~$u'$, and $u'$ is the only target vertex 
(note that all vertices in the first two rows
become unreachable and can be safely deleted). Clearly, 
$\val_{\HR}(v) = 1$, but Player~$\Box$ has no optimal strategy.

Similarly, an optimal strategy for Player~$\Diamond$ may not exist
even if $V_\Box = \emptyset$ \cite{BBFK:BPA-games-reachability-IC,%
Kucera:games-chapter}. To see this, consider the vertex~$u$ 
of Fig.~\ref{fig-determ}, where $t$ is the only target vertex and 
the depicted game is modified by
redirecting the only outgoing transition of~$p$ back to~$u$ (this
makes all vertices in the last two rows unreachable).
We have that $\val_{\HR}(u) = 0$, but Player~$\Diamond$
has no optimal strategy.

One may be also tempted to think that if Player~$\Box$ (or Player~$\Diamond$)
has \emph{some} optimal strategy, then he also has an optimal MD strategy.
However, optimal strategies generally require \emph{infinite memory} even for
reachability objectives (this holds for both players). Since the
corresponding counterexamples are not completely trivial, we refer to 
\cite{Kucera:games-chapter} for details. Interestingly, an optimal strategy
for Player~$\Box$ may also require \emph{randomization}. Consider
the vertex $v$ of Fig.~\ref{fig:MRopt}. 
Let $\sigma^* \in \MR_{\Box}$ be a strategy
selecting $v \tran{} q_n$ with probability~$1/2^n$. Since 
$V_{\Diamond} = \emptyset$, we have that
$\inf_{\pi \in \HR_{\Diamond}} \Exp^{\sigma^*,\pi}_v[\Acc]= \infty = \val_{\HR}(v)$.
However, for every $\sigma \in \HD_{\Box}$ we have that 
$\inf_{\pi \in \HR_{\Diamond}} \Exp^{\sigma,\pi}_v[\Acc] < \infty$.
 
\tikzstyle{dummy}=[font=\scriptsize,inner sep=5pt]
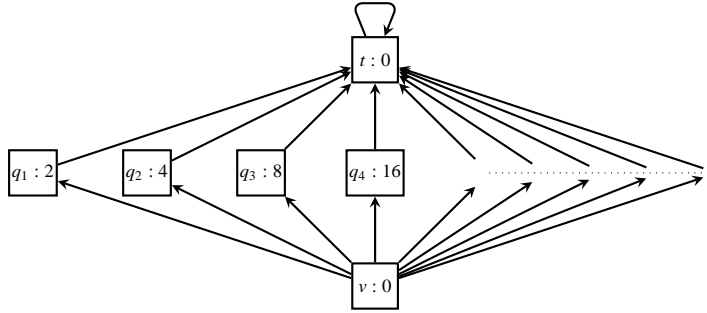
\begin{figure}[t]
 \centering
 \begin{tikzpicture}[x=1.5cm,y=1.5cm,font=\scriptsize]
    \node (v) at (0,0) [max] {$v: 0$};
  \foreach \x/\n/\r in {-3/1/2,-2/2/4,-1/3/8,0/4/16}{\node (q\n) at (\x,1) [max,inner sep=1pt,aspect=1] {$q_{\n}: \r$}; \draw [tran] (v) to (q\n);}
  \foreach \x/\n in {1/1,1.5/2,2/3,2.5/4,3/5}{\node (d\n) at (\x,1) [dummy] {}; \draw [tran] (v) to (d\n);}
  \node (t) at (0,2) [max] {$t: 0$};
  \foreach \n in {1,2,3,4}{\draw [tran] (q\n) to (t);}
  \foreach \x in {1,2,3,4,5}{\draw [tran] (d\x) to (t);}
  \draw [-,dotted] (1,1) to (3,1);
 \draw [tran,rounded corners] (t) -- +(-.2,0.5) -- +(.2,0.5) -- (t); 
 \end{tikzpicture}
 \caption{Player~$\Box$ has an $\MR$-optimal strategy in~$v$, 
   but no $\HD$-optimal strategy in~$v$. All vertices are labelled by pairs
   of the form \emph{vertex name}:\emph{reward}.}
 \label{fig:MRopt}
\end{figure}

For $\Diamond$-finitely-branching games, the situation is somewhat different,
as our second main theorem reveals.

\begin{theorem}
  \label{thm:finite-branching-determinacy}
  Let $\G$ be the class of all $\Diamond$-finitely-branching games. Then
  $\G$ is \mbox{HR-determined}, HD-determined, MR-determined, and MD-determined,
  and for every vertex $v$ of every $G \in \G$ and every reward 
  function $r$ we have that 
  \[
     \val_{\HR}(v) = \val_{\HD}(v) = \val_{\MR}(v) = \val_{\MD}(v) \,.
  \]
  Further, for every $G \in \G$ there exists a MD strategy for 
  Player~$\Diamond$ which is optimal in every vertex of~$G$.
\end{theorem}
An optimal strategy for Player~$\Box$ may not exist in 
$\Diamond$-finitely-branching
games, and even if it does exist, it may require infinite 
memory~\cite{Kucera:games-chapter}.

Theorems~\ref{thm:general-determinacy} 
and~\ref{thm:finite-branching-determinacy} are proven by a sequence of
lemmas presented below. 
For the rest of this section, we fix a stochastic game
$G=(V,\tran{},(V_{\Box},V_\Diamond,V_{\bigcirc}),\Prob)$ and a
reward function $r \colon V \rightarrow \Rset^{\geq 0}$. 
We start with the first part of Theorem~\ref{thm:general-determinacy}~(a), 
i.e., we show that every vertex has a HR-value. This is achieved
by defining a suitable Bellman operator $L$ and proving that the 
least fixed-point of $L$ is the tuple of all HR-values. More
precisely, let $L \colon (\Rset^{\geq 0}_\infty)^{V}
\rightarrow (\Rset^{\geq 0}_\infty)^{V}$, where $\vec{y} = L(\vec{x})$
is defined as follows: 
\[
 \vec{y}_v = \begin{cases}
                r(v) + \sup_{v \tran{} v'} \vec{x}_{v'} & \text{if } v \in V_{\Box}\\
		r(v) + \inf_{v \tran{} v'} \vec{x}_{v'} & \text{if } v \in V_{\Diamond}\\
		r(v) + \sum_{v \tran{} v'} \vec{x}_{v'}\cdot \Prob(v)(v,v') & \text{if } v \in V_{\bigcirc}.
                \end{cases}
\]
A proof of the following lemma can be found in Appendix~\ref{app-proofs}.
Some parts of this proof are subtle, and we also need to make several 
observations that are useful for proving the other results.

\begin{lemma}
\label{lem:hrdet}
The operator $L$ has the least fixed point $\vec{K}$ (w.r.t.{} $\vleq$)
and for every $v \in V$ we have that
\[
  \vec{K}_v \quad = \quad \adjustlimits \sup_{\sigma \in
  \HR_{\Box}} \inf_{\pi\in \HR_{\Diamond}} \Exp^{\sigma,\pi}_{v}[\Acc]
  \quad = \quad \adjustlimits\inf_{\pi\in \HR_{\Diamond}} \sup_{\sigma \in
  \HR_{\Box}} \Exp^{\sigma,\pi}_{v}[\Acc] \quad = \quad \val_{\HR}(v).
\]
Moreover, for every $\eps>0$ there is  $\pi_{\eps} \in \HD_{\Diamond}$
such that for every $v \in V$ we have that
$\sup_{\sigma \in \HR_{\Box}}\Exp^{\sigma,\pi_{\eps}}_v \leq \val_{\HR}(v)\oplus\eps$.
\end{lemma}

To complete our proof of Theorem~\ref{thm:general-determinacy}~(a),
we need to show the existence of a \mbox{HD-value} in every vertex,
and demonstrate that HR and HD values are equal.
Due to Lemma~\ref{lem:hrdet}, for every $\varepsilon >0$ 
there is $\pi_\varepsilon \in \HD_{\Diamond}$ such that $\pi_\varepsilon$
is $\varepsilon$-$\HR$-optimal in every vertex.
Hence, it suffices to show the same for Player~$\Box$.
The following lemma is proved in Appendix~\ref{sec:app-HD}.

\begin{lemma}
\label{lem:eps-HD-opt}
  For every $\eps>0$, there is $\sigma_{\eps} \in \HD_{\Box}$ such that
  $\sigma_{\eps}$ is $\eps$-$\HR$-optimal in every vertex.
\end{lemma}

The next lemma proves Item~(b) of Theorem~\ref{thm:general-determinacy}.

\begin{lemma}
  Consider the vertex $v$ of the game shown in Fig.~\ref{fig-determ},
  where $t$ is the only target vertex and all probability distributions
  assigned to stochastic states are uniform. Then 
  \begin{itemize}
  \item[(a)] $\sup_{\sigma \in \MD_{\Box}} \inf_{\pi \in \MD_{\Diamond}} 
            \Exp^{\sigma,\pi}_v[\Reach] \ = \
         \sup_{\sigma \in \MR_{\Box}} \inf_{\pi \in \MR_{\Diamond}} 
            \Exp^{\sigma,\pi}_v[\Reach] \ = \ 0$;
  \item[(b)] $\inf_{\pi \in \MD_{\Diamond}} \sup_{\sigma \in \MD_{\Box}} 
            \Exp^{\sigma,\pi}_v[\Reach] \ = \ 
         \inf_{\pi \in \MR_{\Diamond}} \sup_{\sigma \in \MR_{\Box}} 
            \Exp^{\sigma,\pi}_v[\Reach] \ = \ 1$.
  \end{itemize}
\end{lemma}
\begin{proof}
  We start by proving item~(a) for MD strategies. Let $\sigma^* \in \MD_{\Box}$.
  We show that 
  $\inf_{\pi \in \MD_{\Diamond}} \Exp^{\sigma^*,\pi}_v[\Reach] = 0$. Let us fix
  an arbitrarily small $\varepsilon > 0$. We show that there is a suitable 
  $\pi^* \in \MD_{\Diamond}$ such that 
  $\Exp^{\sigma^*,\pi^*}_v[\Reach] \leq \varepsilon$.
  If the probability of reaching the vertex
  $u$ from $v$ under the strategy $\sigma^*$ is at most $\varepsilon$,
  we are done. Otherwise, let $p_s$ be the probability of visiting
  the vertex $s$ from $v$ under the strategy $\sigma$ \emph{without}
  passing through the vertex $u$. Note that $p_s >0$ and $p_s$ does 
  not depend on the strategy chosen by Player~$\Diamond$. The strategy
  $\pi^*$ selects a suitable successor of $u$ such that the probability 
  $p_t$ of visiting the vertex $t$ from $u$ without passing through 
  the vertex~$v$ satisfies $p_t/p_s < \varepsilon$ (note that $p_t$ 
  can be arbitrarily small but positive). Then 
  \[
     \Exp^{\sigma^*,\pi^*}_v[\Reach] \quad \leq \quad
     \sum_{i=1}^\infty (1-p_s)^i p_t  \quad = \quad
     \frac{(1-p_s) p_t}{p_s} \quad \leq \quad \varepsilon
  \]  
  For MR strategies, the argument is the same.

  Item~(b) is proven similarly. We show that for all $\pi^* \in \MD_{\Diamond}$
  and $0 < \varepsilon < 1$ there exists a suitable 
  $\sigma^* \in \MD_{\Box}$ such that 
  $\Exp^{\sigma^*,\pi^*}_v[\Reach] \geq 1 - \varepsilon$. Let $p_t$ be the
  probability of visiting $t$ from $u$ without passing through 
  the vertex~$v$ under the strategy $\pi^*$. We choose the strategy 
  $\sigma^*$ so that the probability $p_s$ of visiting the vertex 
  $s$ from $v$ without passing through the vertex~$u$ satisfies 
  $p_s/p_t < \varepsilon$. Note almost all runs initiated in $v$
  eventually visit either $s$ or~$t$ under $(\sigma^*,\pi^*)$. Since 
  the probability of visiting~$s$
  is bounded by $\varepsilon$ (the computation is similar to the one of
  item~(a)), we obtain $\Exp^{\sigma^*,\pi^*}_v[\Reach] \geq 1 - \varepsilon$.
  For MR strategies, the proof is almost the same.
\qed
\end{proof}

\begin{figure}[t]
\centering
\begin{tikzpicture}[x=1.5cm,y=1.5cm,font=\scriptsize]
\node (N00) at (0,0)   [max] {$v$};
\node (N01) at (0,1)   [ran] {$s$};
\node (N02) at (0,2)   [min] {$u$};
\node (N03) at (0,3)   [ran,double] {$t$};
\node (N04) at (0,4)   [ran] {$p$};
\foreach \x in {1,2,3,4,5}{%
   \node (N\x 0) at (\x,0)   [max] {};
   \node (N\x 1) at (\x,1)   [ran] {};
   \node (N\x 2) at (\x,2)   [ran] {};
   \node (N\x 3) at (\x,3)   [ran] {};
   \node (N\x 4) at (\x,4)   [ran] {};
   \draw [tran] (N\x 0) to (N\x 1); 
   \draw [tran] (N\x 1) to (N\x 2); 
   \draw [tran] (N\x 3) to (N\x 4); 
}
\foreach \y in {0,1,2,3,4}{%
   \node (N6\y) at (6,\y)   [max,draw=none] {};
   \draw [thick,dotted] (N6\y) -- +(1,0);
}
\foreach \x/\xx in {0/1,1/2,2/3,3/4,4/5,5/6}{%
   \draw [tran] (N\x 0) to (N\xx 0);
   \foreach \y in {1,2,3,4}{%
      \draw [tran] (N\xx\y) to (N\x\y);
   } 
}
\foreach \x in {0,1,2,3,4,5,6}{%
   \draw [tran,rounded corners] (N02) -- +(0,0.5) -- +(\x,0.5) -- (N\x 3);
}
\draw [tran,rounded corners] (N01) -- +(-.5,0.2) -- +(-.5,-.2) -- (N01); 
\draw [tran,rounded corners] (N03) -- +(-.5,0.2) -- +(-.5,-.2) -- (N03); 
\draw [tran,rounded corners] (N04) -- +(-1,0) -- +(-1,-4) -- (N00);

\end{tikzpicture}
\caption{A game whose vertex $v$ has neither MD-value nor MR-value.}
\label{fig-determ}
\end{figure}
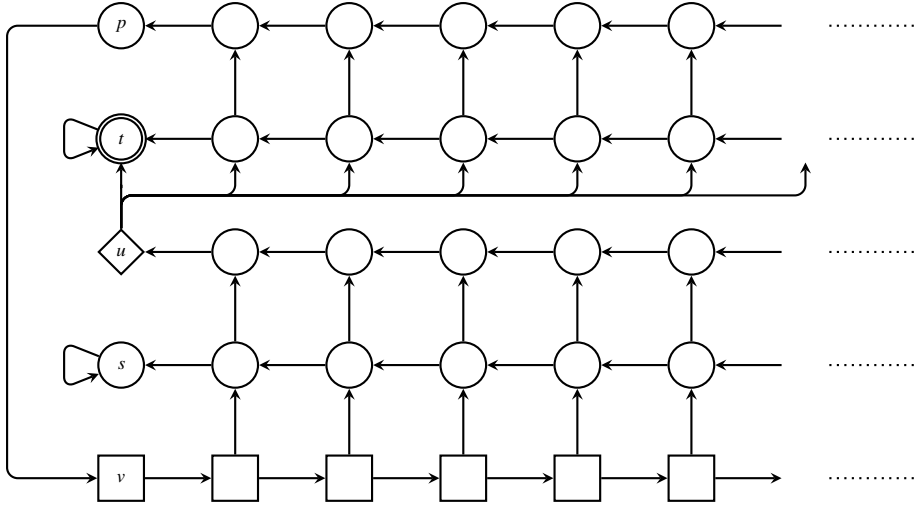




We continue by proving Theorem~\ref{thm:finite-branching-determinacy}. 
This theorem follows immediately from Lemma~\ref{lem:hrdet} and the 
following proposition:
\begin{proposition}\label{prop:finb-MD}
  If $G$ is $\Diamond$-finitely-branching, then 
\begin{enumerate}
\item \label{prop:finb-MD-box} for all $v\in V$ and $\eps>0$, 
   there is $\sigma_{\eps} \in \MD_{\Box}$ such that $\sigma_{\eps}$
   is $\eps$-$\HR$-optimal in~$v$;
\item \label{prop:finb-MD-diam} there is $\pi \in \MD_{\Diamond}$ such  
   that $\pi$ is $\HR$-optimal in every vertex.
\end{enumerate}
\end{proposition}

As an immediate corollary to  Proposition~\ref{prop:finb-MD}, 
we obtain the following result:
\begin{corollary}
\label{cor-MD-optimal}
If $G$ is $\Diamond$-finitely-branching, $V_\Box$ is finite, and
every vertex of $V_\Box$ has finitely many successors,
then there is $\sigma \in \MD_{\Box}$ such that $\sigma$ is 
\mbox{$\HR$-optimal} in every vertex.
\end{corollary}
\begin{proof}
  Due to Proposition~\ref{prop:finb-MD}, for every vertex $v$ and 
  every $\eps > 0$, there is $\sigma_{\eps} \in \MD_{\Box}$ such 
  that $\sigma_{\eps}$ is $\eps$-$\HR$-optimal in~$v$. Since 
  $V_\Box$ is finite and every vertex of $V_\Box$ has only finitely many
  successors, there are only finitely many $\MD$-strategies 
  for Player~$\Box$. Hence, there is a MD~strategy $\sigma$ that is
  $\eps$-$\HR$-optimal in $v$ for infinitely many $\eps$ from
  the set $\{1,1/2,1/4,\dots\}$. Such a strategy is clearly
  $\HR$-optimal in~$v$. Note that $\sigma$ is \mbox{$\HR$-optimal} in
  every vertex which can be reached from $v$ under $\sigma$ and some
  strategy $\pi$ for Player~$\Diamond$. For the remaining vertices,
  we can repeat the argument, and thus eventually produce a MD strategy
  that is $\HR$-optimal in every vertex.
\qed
\end{proof}

Hence, if all non-stochastic vertices have finitely many successors and
$V_\Box$ is finite, then both players have $\HR$-optimal MD~strategies.
This can be seen as a (tight) generalization of the corresponding
result for finite-state games~\cite{FV:book}. 

The rest of this section is devoted to a proof of
Proposition~\ref{prop:finb-MD}. We start with Item~\ref{prop:finb-MD-box}.
The strategy $\sigma_{\varepsilon}$ is constructed by employing 
discounting. Assume, w.l.o.g., that rewards are bounded by~$1$
(if they are not, we may split
every state $v$ with a reward $r(v)$ into a sequence of $\lceil
r(v)\rceil$ states, each with the reward $r(v)/\lceil r(v)\rceil$).
Given $\lambda\in (0,1)$, define $\Acc^{\lambda} : \run \rightarrow
\Rset^{\geq 0}$ to be a function which to every run~$\omega$ assigns
$\Acc^{\lambda}(\omega) = \sum_{i=0}^{\infty} \lambda^i \cdot
r(\omega(i))$.
\begin{lemma}\label{lem:discount}
For $\lambda$ sufficiently close to one we have that 
\[
\adjustlimits\sup_{\sigma\in\HR_{\Box}} \inf_{\pi\in\HR_{\Diamond}} \Exp^{\sigma,\pi}_v(\Acc^{\lambda})
  \quad \geq \quad
\val_{\HR}(v) \ominus \frac{\varepsilon}{2}
\]
\end{lemma}
\begin{proof}
  We show that for every $\varepsilon>0$ there is $n\geq 0$ such
  that the expected reward that Player~$\Box$ may accumulate up to~$n$
  steps is $\varepsilon$-close to $\val_{\HR}(v)$  no matter what
  Player~$\Diamond$ is doing.  Formally, define $\Acc_k : \run
  \rightarrow \Rset^{\geq 0}$ to be a function which to every
  run~$\omega$ assigns $\Acc_k(\omega) = \sum_{i=0}^{k} r(\omega(i))$.
  The following lemma is proved in Appendix~\ref{app:nstep}.
\begin{lemma}
\label{lem-max-bound} 
If $G$ is $\Diamond$-finitely-branching, then for every 
$v \in V$ there is $n\in \Nset$ such that
\[ 
  \adjustlimits\sup_{\sigma\in\HR_{\Box}} \inf_{\pi\in\HR_{\Diamond}} \Exp^{\sigma,\pi}_v(\Acc_n) 
  \quad > \quad \val_{\HR}(v) \ominus \frac{\varepsilon}{4}
\]
\end{lemma}
Clearly, if $\lambda$ is close to one, then for every run $\omega$ we
have that
\[
  \Acc^{\lambda}(\omega)\quad \geq \quad \Acc_n(\omega) -
 \frac{\varepsilon}{4}
\]
Thus,
\[
\adjustlimits\sup_{\sigma\in\HR_{\Box}} \inf_{\pi\in\HR_{\Diamond}} \Exp^{\sigma,\pi}_v(\Acc^{\lambda})
\quad \geq \quad 
\adjustlimits\sup_{\sigma\in\HR_{\Box}} \inf_{\pi\in\HR_{\Diamond}} \Exp^{\sigma,\pi}_v(\Acc_n)-\frac{\varepsilon}{4}
\quad \geq \quad 
\val_{\HR}(v) \ominus \frac{\varepsilon}{2}
\]
This proves Lemma~\ref{lem:discount}. \qed
\end{proof}
So, it suffices to find a MD strategy $\sigma_{\varepsilon}$
satisfying 
\[
 \inf_{\pi\in\HR_{\Diamond}} \Exp^{\sigma_{\varepsilon},\pi}_v(\Acc^{\lambda})
  \quad \geq \quad
 \adjustlimits\sup_{\sigma\in\HR_{\Box}}
 \inf_{\pi\in\HR_{\Diamond}} \Exp^{\sigma,\pi}_v(\Acc^{\lambda})-\frac{\varepsilon}{2} \,.
\]
We define such a strategy as follows.  Let us fix some $\ell \in \Nset$ 
satisfying
\[
  \frac{\lambda^{\ell}}{1-\lambda} \cdot \max_{v\in V}
  r(v) \quad < \quad \frac{\varepsilon}{8} \,.
\]
Intuitively, the discounted reward accumulated after $\ell$ steps 
can be at most $\frac{\varepsilon}{8}$.  In a given vertex 
$v\in V_{\Box}$, the strategy
$\sigma_{\varepsilon}$ chooses a fixed successor vertex~$u$ satisfying
\[
\adjustlimits\sup_{\sigma\in\HR_{\Box}}\inf_{\pi\in\HR_{\Diamond}} \Exp^{\sigma,\pi}_u(\Acc^{\lambda})\quad \geq\quad \adjustlimits\sup_{v\tran{} u'}\sup_{\sigma\in\HR_{\Box}}\inf_{\pi\in\HR_{\Diamond}} \Exp^{\sigma,\pi}_{u'}(\Acc^{\lambda}) - \frac{\varepsilon}{\ell\cdot 4}
\]
Now we show that 
\[
\inf_{\pi\in\HR_{\Diamond}} \Exp^{\sigma_{\varepsilon},\pi}_v(\Acc^{\lambda})\quad \geq \quad 
\adjustlimits\sup_{\sigma\in\HR_{\Box}} \inf_{\pi\in\HR_{\Diamond}} \Exp^{\sigma,\pi}_v(\Acc^{\lambda})-\frac{\varepsilon}{2}\, .
\]
which finishes the proof of Item~\ref{prop:finb-MD-box} of 
Proposition~\ref{prop:finb-MD}.

For every ${k}\in \Nset$ we denote by $\sigma_{{k}}$ a strategy for
Player $\Box$ defined as follows: For the first ${k}$ steps the
strategy makes the same choices as $\sigma_{\varepsilon}$, i.e.,
chooses, in each state $v\in V_{\Box}$, a next state $u$ satisfying
\[
\adjustlimits\sup_{\sigma\in\HR_{\Box}}\inf_{\pi\in\HR_{\Diamond}} \Exp^{\sigma,\pi}_u(\Acc^{\lambda})\quad \geq\quad \sup_{v\tran{} u'} \sup_{\sigma\in\HR_{\Box}}\inf_{\pi\in\HR_{\Diamond}} \Exp^{\sigma,\pi}_{u'}(\Acc^{\lambda}) - \frac{\varepsilon}{{k}\cdot 4}
\]
From $k{+}1$-st step on, say in a state $u$, the strategy follows some strategy $\zeta$ satisfying
\[
\inf_{\pi\in\HR_{\Diamond}} \Exp^{\zeta,\pi}_u(\Acc^{\lambda}) 
   \quad \geq \quad \sup_{\sigma\in\HR_{\Box}} \inf_{\pi\in\HR_{\Diamond}} \Exp^{\sigma,\pi}_u(\Acc^{\lambda})-\frac{\varepsilon}{8} 
\]
A simple induction reveals that $\sigma_{{k}}$ satisfies
\begin{equation}\label{eq:k-step-disc}
\inf_{\pi\in\HR_{\Diamond}} \Exp^{\sigma_{{k}},\pi}_v(\Acc^{\lambda}) 
   \quad \geq \quad \sup_{\sigma\in\HR_{\Box}} \inf_{\pi\in\HR_{\Diamond}} \Exp^{\sigma,\pi}_v(\Acc^{\lambda})-\frac{3\varepsilon}{8} 
\end{equation}
(Intuitively, the error of each of the first ${k}$ steps is at most $\frac{\varepsilon}{{k}\cdot 4}$ and
thus the total error of the first ${k}$ steps is at most ${k} \cdot \frac{\varepsilon}{{k}\cdot 4}=\frac{\varepsilon}{4}$. The rest has the error at most $\frac{\varepsilon}{8}$ and thus the total error is at most $\frac{3\varepsilon}{8}$.)

We consider $k=\ell$ (recall that $\frac{\lambda^{\ell}}{1-\lambda} \cdot \max_{v\in V} r(v)<\frac{\varepsilon}{8}$).
Then
\[
\inf_{\pi\in\HR_{\Diamond}}\Exp^{\sigma_{\varepsilon},\pi}_v(\Acc^{\lambda}) \quad \geq \quad
   \inf_{\pi\in\HR_{\Diamond}} \Exp^{\sigma_{k},\pi}_v(\Acc^{\lambda})-\frac{\varepsilon}{8} 
   \quad \geq \quad \sup_{\sigma\in\HR_{\Box}} \inf_{\pi\in\HR_{\Diamond}} \Exp^{\sigma,\pi}_v(\Acc^{\lambda})-\frac{\varepsilon}{2} 
\]
Here the first equality follows from the fact that $\sigma_k$ behaves similarly to $\sigma_{\varepsilon}$ on the first $k=\ell$ steps and the discounted reward accumulated after $k$ steps is at most $\frac{\varepsilon}{8}$. The second inequality follows from Equation~(\ref{eq:k-step-disc}).

It remains to prove Item~\ref{prop:finb-MD-diam} of 
Proposition~\ref{prop:finb-MD}. The MD strategy $\pi$ can be easily 
constructed as follows: In every state $v\in V_{\Diamond}$, the 
strategy $\pi$ chooses a successor $u$ minimizing $\val_{\HR}(u)$ 
among all successors of $v$. We show in Appendix~\ref{app:MD-opt-min} that this 
is indeed an optimal strategy.

\section{Conclusions}
\label{sec-concl}

We have considered infinite-state stochastic games with the total
accumulated reward objective, and clarified the determinacy questions
for the HR, HD, MR, and MD strategy types. Our results are almost 
complete. One natural question which remains open is whether
Player~$\Box$ needs memory to play $\eps$-$\HR$-optimally in general
games (it follows from the previous works, e.g., 
\cite{BBFK:BPA-games-reachability-IC,Kucera:games-chapter}, 
that $\eps$-$\HR$-optimal strategies for Player~$\Diamond$ require 
infinite memory in general).


%
%

\bibliography{str-long,concur}

\newpage
\appendix

\noindent
\begin{center}
{\huge{\bf Technical Appendix}}
\end{center}
\section{Proof of Lemma \ref{lem:hrdet}}
\label{app-proofs}

\begin{reflemma}{lem:hrdet}
The operator $L$ has the least fixed point $\vec{K}$ (w.r.t.{} $\vleq$)
and for every $v \in V$ we have that
\[
  \vec{K}_v \quad = \quad \adjustlimits \sup_{\sigma \in
  \HR_{\Box}} \inf_{\pi\in \HR_{\Diamond}} \Exp^{\sigma,\pi}_{v}[\Acc]
  \quad = \quad \adjustlimits\inf_{\pi\in \HR_{\Diamond}} \sup_{\sigma \in
  \HR_{\Box}} \Exp^{\sigma,\pi}_{v}[\Acc] \quad = \quad \val_{\HR}(v).
\]
Moreover, for every $\eps>0$ there is  $\pi_{\eps} \in \HD_{\Diamond}$
such that for every $v \in V$ we have that
$\sup_{\sigma \in \HR_{\Box}}\Exp^{\sigma,\pi_{\eps}}_v \leq \val_{\HR}(v)\oplus\eps$.
\end{reflemma}
The partially ordered set $((\Rni)^{V},\vleq)$, where $\vleq$ is a standard componentwise ordering, is a complete lattice. Moreover, from the definition of $L$ we can easily see that $L$ is monotonic, i.e. $L(\vec{x})\vleq L(\vec{x'})$ whenever $\vec{x} \vleq \vec{x'}$. Thus, by the Knaster-Tarski theorem the operator $L$ has the least fixed point, which we denote by $\vec{K}$.

In order to prove that $\vec{K}_v = \val_{\HR}(v)$ for every $v \in V$, it suffices to prove the following:
\begin{equation}
\label{eq:det1}
 \forall v \in V: \quad \vec{K}_{v} \leq \adjustlimits \sup_{\sigma \in \HR_{\Box}} \inf_{\pi\in \HR_{\Diamond}} \Exp^{\sigma,\pi}_{v}(\Acc) \leq  \adjustlimits\inf_{\pi\in \HR_{\Diamond}} \sup_{\sigma \in \HR_{\Box}} \Exp^{\sigma,\pi}_{v}(\Acc) \leq \vec{K}_{v}.
\end{equation}
The second inequality holds trivially, so it suffices to prove the remaining ones.

To prove the first inequality, it suffices to show that the vector $\vec{S} \in (\Rni)^{V}$ defined by $\vec{S}_v=\sup_{\sigma \in \HR_{\Box}} \inf_{\pi\in \HR_{\Diamond}} \Exp^{\sigma,\pi}_{v}(\Acc)$ is a fixed point of $L$. Since $\vec{K}$ is the least fixed point of $L$, the inequality then follows. So let $v \in V$ be arbitrary. We will show that $L(\vec{S})_v = \vec{S}_v$.

If $v \in V_{\Box}$, then we have to show that $$L(\vec{S})_v = r(u) + \sup_{v\tran{}v'}\adjustlimits\sup_{\sigma\in\HR_{\Box}}\inf_{\pi\in\HR_{\Diamond}}\Exp^{\sigma,\pi}_{v'}(\Acc) = \adjustlimits\sup_{\sigma\in\HR_{\Box}}\inf_{\pi\in\HR_{\Diamond}}\Exp^{\sigma,\pi}_{v}(\Acc) = \vec{S}_v.$$ Assume, for the sake of contradiction, that the equality does not hold, i.e. that either $L(\vec{S})_v < \vec{S}_v$ or $L(\vec{S})_v > \vec{S}_v$. If $L(\vec{S})_v > \vec{S}_v$, then there is a transition $v\tran{} v'$ and a strategy $\sigma' \in \HR_{\Box}$ such that $r(u) + \inf_{\pi\in\HR_{\Diamond}}\Exp^{\sigma',\pi}_{v'}(\Acc) > \sup_{\sigma\in\HR_{\Box}}\inf_{\pi\in\HR_{\Diamond}}\Exp^{\sigma,\pi}_{v}(\Acc)$. If we denote by $\sigma''$ the strategy that moves from the initial vertex $v$ to $v'$ with probability 1 and then starts to behave exactly like the strategy $\sigma'$, then we obtain
\[
 \inf_{\pi\in\HR_{\Diamond}}\Exp^{\sigma'',\pi}_{v}(\Acc) = r(u) + \inf_{\pi\in\HR_{\Diamond}}\Exp^{\sigma',\pi}_{v'}(\Acc) > \adjustlimits\sup_{\sigma\in\HR_{\Box}}\inf_{\pi\in\HR_{\Diamond}}\Exp^{\sigma,\pi}_{v}(\Acc) \geq \inf_{\pi\in\HR_{\Diamond}}\Exp^{\sigma'',\pi}_{v}(\Acc),
\]
a contradiction. So assume that $L(\vec{S})_v < \vec{S}_v$. 
Then there is some $\delta>0$ and some function $f\colon \HR_{\Box}\times V \rightarrow \HR_{\Diamond}$ such that for every transition $v \tran{} v'$ and every $\sigma \in \HR_{\Box}$ we have $r(u) + \Exp^{\sigma,f(\sigma,v')}_{v'} < \vec{S}_v\ominus\delta$. For any strategy $\sigma$ we denote by $p_{\sigma}^{v'}$ the probability the strategy $\sigma$ assigns to transition $v \rightarrow v'$ in a game starting in $v$. Then we can write
\begin{align*}
 \adjustlimits\sup_{\sigma\in\HR_{\Box}}\inf_{\pi\in\HR_{\Diamond}}\Exp^{\sigma,\pi}_{v}(\Acc) &= r(u) +  \adjustlimits\sup_{\sigma\in\HR_{\Box}}\inf_{\pi\in\HR_{\Diamond}}\sum_{v\tran{}v'}p_{\sigma}^{v'}\cdot\Exp^{\sigma,\pi}_{v'}(\Acc)\\ &\leq r(u) + \sup_{\sigma\in\HR_{\Box}}\sum_{v\tran{}v'}p_{\sigma}^{v'}\cdot\Exp^{\sigma,f(\sigma,v')}_{v'}(\Acc) < \vec{S}_v\ominus\delta \\&\leq \vec{S}_v=\sup_{\sigma\in\HR_{\Box}}\inf_{\pi\in\HR_{\Diamond}}\Exp^{\sigma,\pi}_{v}(\Acc),
\end{align*}
again a contradiction.

For $v \in V_{\Diamond}$ the proof is dual to the proof for $v \in V_{\Box}$, so we omit it. Finally, for $v \in V_{\bigcirc}$ we have
\begin{align*}
 L(\vec{S}_v) &= r(u) +\sum_{v \tran{} v'}\Prob(v)(v,v')\cdot\left(\adjustlimits\sup_{\sigma\in\HR_{\Box}}\inf_{\pi\in\HR_{\Diamond}}\Exp^{\sigma,\pi}_{v'}(\Acc)\right)\\ &=\adjustlimits\sup_{\sigma\in\HR_{\Box}}\inf_{\pi\in\HR_{\Diamond}}\left( r(u) +\sum_{v \tran{} v'}\Prob(v)(v,v')\cdot\Exp^{\sigma,\pi}_{v'}(\Acc)\right) = \adjustlimits\sup_{\sigma\in\HR_{\Box}}\inf_{\pi\in\HR_{\Diamond}}\Exp^{\sigma,\pi}_{v}(\Acc)= \vec{S}_v.
\end{align*}
 This concludes the proof that $\vec{S}$ is a fixed point of $L$ and thus also the proof of the first inequality in \eqref{eq:det1}.

 It remains to prove the third inequality in \eqref{eq:det1}. To this end we prove that for every $\eps>0$ there is a strategy $\pi_{\eps}\in\HD_{\Diamond}$ such that for every $v\in V$ we have $\sup_{\sigma\in\HR_{\Box}}\Exp^{\sigma,\pi_{\eps}}_v \leq \vec{K}_v + \eps$. Note that this will also prove the second part of the lemma.

 If $\vec{K}_v = \infty$, then the desired inequality holds trivially for any strategy of player $\Diamond$ (and particularly for every $\pi\in\HD_{\Diamond}$). So assume that $\vec{K}_v$ is finite and fix arbitrary $\eps>0$. We define the strategy $\pi_{\eps}$ as follows: let $wu$ be any finite path with $u \in V_{\Diamond}$. Since $\vec{K}$ is a fixed point of $L$, there must be a successor $u'$ of $u$ such that $r(u) + \vec{K}_{u'} \leq \vec{K}_{u} + \eps/2^{|wu|+1}$. We set $\pi_{\eps}(w)$ to be a Dirac distribution that selects the transition $u \tran{} u'$ with probability 1.

We will now prove the following lemma, that not only shows that the strategy $\pi_{\eps}$ has the desired property, but it will also be useful later.
\begin{lemma}
 \label{lem:eps-min}
Let $\eps\geq 0$ be arbitrary and let $\pi_{\eps}$ be any deterministic strategy of player $\Diamond$ that has the following property: for every finite path $wu$ starting in $v$ and ending in $u\in V_{\Diamond}$, the transition $u \tran{} u'$ selected by $\pi_{\eps}(wu)$ satisfies $r(u) + \vec{K}_{u'} \leq \vec{K}_{u} + \eps/2^{|wu|+1}$. Then $\sup_{\sigma\in\HR_{\Box}}\Exp^{\sigma,\pi_{\eps}}_v(\Acc) \leq \vec{K}_v + \eps$.
\end{lemma}
\begin{proof}

  We will prove that for every $v$, every $n \in \Nset_0$ and every strategy $\sigma$ of player $\Box$ we have $\Exp^{\sigma,\pi_{\eps}}_v (\sum_{i=0}^{n}\omega(i)) \leq \vec{K}_v + \eps$. By the monotone convergence theorem this means that $\Exp^{\sigma,\pi_\eps}_{v}(\Acc)\leq \vec{K}_v + \eps$ for every $\sigma$, and thus also $\sup_{\sigma\in\HR_{\Box}}\Exp^{\sigma,\pi_\eps}_{v}(\Acc)\leq \vec{K}_v + \eps$.

 So let us fix arbitrary $v$, $n$ and $\sigma$. Recall that $\Exp^{\sigma,\pi}_v [X| Y]$ denotes the conditional expectation of random variable $X$ given the event $Y$. We show that for every $0 \leq k \leq n$ and every finite path $w = v_0,\dots,v_k$ we have \[\Exp^{\sigma,\pi_\eps}_{v}[\sum_{i=k}^{n}r(\omega(i))~|~\run(w)] \leq \vec{K}_{v_k} + \sum_{i=k}^{n} \eps/2^{k+1}.\] In particular, this means that $\Exp^{\sigma,\pi_{\eps}}_v (\sum_{i=0}^{n}\omega(i)) = \E^{\sigma,\pi_\eps}_{v}[\sum_{i=0}^{n}r(\omega(i))~|~\run(v)] \leq \vec{K}_v + \eps$.
 
We proceed by downward induction on $k$. If $n=k$, then we trivially have
\begin{align*}
 \Exp^{\sigma,\pi_\eps}_{v}[\sum_{i=k}^{n}r(\omega(i))~|~ \run(w)] = r(v_k) \leq L(\vec{K})_{v_k} =\vec{K}_{v_k},
\end{align*}
where the inequality follows from the definition of $L$.

Now suppose that $k<n$. We distinguish two cases. If $v_k \in V_{\Diamond}$, denote by $u$ the successor of $v_k$ chosen by $\pi_\eps$. Then we have
\begin{align*}
 \Exp^{\sigma,\pi_\eps}_{v}[\sum_{i=k}^{n}r(\omega(i))~|~\run(w)] &= r({v_k}) + \Exp^{\sigma,\pi_\eps}[\sum_{i=k+1}^{n}r(\omega(i))~|~ \run(wu)] \\ &\leq r({v_k}) + \vec{K}_{u} + \sum_{i=k+1}^{n} \eps/2^{i+1} \\ &\leq \vec{K}_{v_k} + \sum_{i=k}^{n} \eps/2^{i+1},
\end{align*}
where the inequality on the second line follows from induction hypothesis and the inequality on the third line follows from the definition of $\pi_{\eps}$.

If $v_k \in V_{\Box} \cup V_{\bigcirc}$, then we can see that $\Exp^{\sigma,\pi_\eps}_{v}[\sum_{i=k}^{n}r(\omega(i))~|~ \run(w)] = \sum_{v_k\tran{}u} p_u \cdot \Exp^{\sigma,\pi_\eps}[\sum_{i=k+1}^{n}r(\omega(i))~|~ \run(wu)]$ for some sequence of real numbers $(p_u)_{v_k \tran{} u}$ s.t. $p_u \geq 0$ for every $u$ and $\sum_{v_k \tran{} u }p_u = 1$. By induction hypothesis we have $\Exp^{\sigma,\pi_\eps}[\sum_{i=k+1}^{n}r(\omega(i))~|~ \run(wu)] \leq K_{u} + \sum_{i=k+1}^{n} \eps/2^{i+1}$ for every $v_k \tran{} u$. Finally, from the definition of $L$ we obtain $\vec{K}_{v_k} = L(\vec{K})_{v_k}\geq \sum_{v_k\tran{} u}p_u\cdot{\vec{K}_u}$ (the inequality can be strict only if $v\in V_{\Box}$). Together, we have
\[
 \Exp^{\sigma,\pi_\eps}_{v}[\sum_{i=k}^{n}r(\omega(i))~|~ \run(w)] \leq \vec{K}_{v_k} + \sum_{i=k+1}^{n} \eps/2^{i+1} < \vec{K}_{v_k} + \sum_{i=k}^{n} \eps/2^{i+1}.
\]
\qed
\end{proof}
This finishes the proof of Lemma \ref{lem:hrdet}.

\section{Proof of Lemma \ref{lem:eps-HD-opt}}
\label{sec:app-HD}
\begin{reflemma}{lem:eps-HD-opt}
 For every $\eps>0$, there is $\sigma_{\eps} \in \HD_{\Box}$ such that
  $\sigma_{\eps}$ is $\eps$-$\HR$-optimal in every vertex.
\end{reflemma}


Let $\eps>0$ be arbitrary. It suffices to fix an arbitrary initial
   vertex $v$, define choices of the strategy $\sigma_\eps$ only on the
   finite paths starting in $v$ and verify, that the resulting strategy
   is $\eps$-$\HR$-optimal in $v$. By repeating this construction for
   every $v\in V$ we obtain a strategy that is $\eps$-$\HR$-optimal in
   every vertex.

For the sake of better readability, we first present the detailed construction of the deterministic $\eps$-$\HR$-optimal strategy $\sigma_{\eps}$ for games in which the $\HR$-value is finite in every vertex. Almost identical construction can be used for games with arbitrary $\HR$-values; there are some subtle technical differences that will be presented in the second part of the proof. 

We already know that the least fixed point $\mathbf{K}$ of the operator $L$ is equal to the vector of $\HR$-values. Moreover, from the standard results of the fixed-point theory (see, e.g., Theorem 5.1 in \cite{CC:fixed-point}) we know that $\mathbf{K}=L^{\alpha}(\mathbf{0})$ for some ordinal number $\alpha$ (where $\mathbf{0}$ is the vector of zeros and where the transfinite iteration of $L$ is defined in a standard way, i.e. we put $L^{\beta}(\mathbf{0})=\sup_{\gamma<\beta}L^{\gamma}(\mathbf{0})$ for every limit ordinal $\beta$). The following lemma is instrumental in the construction of $\sigma_{\eps}$.

\begin{lemma}
\label{lem:ord-labeling}
Let $\eps>0$ be arbitrary.
Denote by $\alpha$ the ordinal number $\alpha$ such that $L^{\alpha}(\mathbf{0})_v=\val_{\HR}(v)$ and denote by $\Ord_{\alpha}$ the set of all ordinal numbers lesser than or equal to $\alpha$. Then there is a labeling function $d\colon \fpath(v) \rightarrow \Ord_{\alpha}$ satisfying the following conditions:
\begin{enumerate}
 \item[(a)] $d(v)=\alpha$.
 \item[(b)] For every $wu\in\fpath(v)$ it holds either $d(w)=0$ or $d(wu)<d(w)$.
 \item[(c)] For every $wu\in\fpath(v)$, we have 
  \[
    L^{d(wu)}(\mathbf{0})_u - \frac{\eps}{2^{|wu|+1}} \leq \begin{cases}
                                                      r(u) + L^{d(wuu')}(\mathbf{0})_{u'}, \text{ for some } u \tran{} u' & \text{ if } u \in V_{\Box}\\
						      r(u) + \inf_{u\tran{} u'} L^{d(wuu')}(\mathbf{0})_{u'} & \text { if } u \in V_{\Diamond}\\
						      r(u) + \sum_{u \tran{} u'}\Prob(u)(u,u')\cdot L^{d(wuu')}(\mathbf{0})_{u'} & \text{ if } u \in V_{\bigcirc}.
                                                      \end{cases}
  \]
\end{enumerate}
\end{lemma}
\begin{proof}
We define the labeling $d$ inductively, proceeding from the shorter paths to the longer ones. Obviously we set $d(v)=\alpha$. Now suppose that $d(wu)$ has already been defined. We will define $d(wuu')$ for all successors $u'$ of $u$ simultaneously. First let us assume that $d(wu)$ is a successor ordinal of the form $\beta+1$. Then it suffices to put $d(wuu')=\beta$ for all successors $u'$ of $u$. From the definition of $L$ we can easily see that for every $\delta>0$ it then holds
\[
    L^{\beta+1}(\mathbf{0})_u - \delta \leq \begin{cases}
                                                      r(u) + L^{\beta}(\mathbf{0})_{u'}, \text{ for some } u \tran{} u' & \text{ if } u \in V_{\Box}\\
						      r(u) + \inf_{u\tran{} u'} L^{\beta}(\mathbf{0})_{u'} & \text { if } u \in V_{\Diamond}\\
						      r(u) + \sum_{u \tran{} u'}\Prob(u)(u,u')\cdot L^{\beta}(\mathbf{0})_{u'} & \text{ if } u \in V_{\bigcirc},
                                                      \end{cases}
  \]

 so in particular the inequality in (c) holds for $wu$. 

Now let us assume that $d(wu)$ is a limit ordinal. Then $L^{d(wu)}(\mathbf{0})_u = \sup_{\gamma < d(wu)}L^{\gamma}(\mathbf{0})_u$. This means that there is $\gamma<d(wu)$ such that $L^{d(wu)}(\mathbf{0})_u - {\eps}/{2^{|wu|+2}}\leq L^{\gamma}(\mathbf{0})_u$. Clearly, we can assume that $\gamma=\beta+1$ fore some ordinal $\beta$. Now we again set $d(wuu')=\beta$ for all successors $u'$ of $u$. Using the argument from the previous paragraph with $\delta={\eps}/{2^{|wu|+2}}$ we obtain
\[
  L^{d(wu)}(\mathbf{0})_u - \frac{\eps}{2^{|wu|+1}} \leq L^{\gamma}(\mathbf{0})_u - \frac{\eps}{2^{|wu|+2}} \leq \begin{cases}
                                                      r(u) + L^{\beta}(\mathbf{0})_{u'}, \text{ for some } u \tran{} u' & \text{ if } u \in V_{\Box}\\
						      r(u) + \inf_{u\tran{} u'} L^{\beta}(\mathbf{0})_{u'} & \text { if } u \in V_{\Diamond}\\
						      r(u) + \sum_{u \tran{} u'}\Prob(u)(u,u')\cdot L^{\beta}(\mathbf{0})_{u'} & \text{ if } u \in V_{\bigcirc},
                                                      \end{cases}
\]
so (c) again holds for $wu$.

Finally, if $d(wu)=0$, then we set $d(wuu')=0$ for all successors $u'$ of $u$. In this way, we eventually define $d(w)$ for every finite path starting in $v$. It is obvious that $d$ satisfies (a)--(c).
\qed
\end{proof}

We use the labeling $d$ provided by the previous lemma to define the $\eps$-$\HR$-optimal $\HD$ strategy $\sigma_{\eps}$ of player $\Box$. For a given finite path $wu$ the strategy $\sigma_{\eps}$ selects a transition $u\tran{}u'$ such that $ L^{d(wu)}(\mathbf{0})_u - {\eps}/{2^{|wu|+1}} \leq r(u) + L^{d(wuu')}(\mathbf{0})_{u'}$. Such a transition always exists due to the previous lemma. We now prove that the strategy $\sigma_{\eps}$ is $\eps$-$\HD$-optimal in $v$. We will actually prove a more general statement, that we will reuse later. 

\begin{lemma}
\label{lem:eps-HD-strat}
 For every run $\omega$ denote by $\tau(\omega)$ the least $k$ such that $d(\omega(0),\dots,\omega(k))=0$ and denote by $S_k^{\tau}$ the random variable defined by $S_k^{\tau}(\omega)=\sum_{i=k}^{\tau(\omega)}r(\omega(i))$. Then the following holds for every $wu\in\fpath(v)$:
\begin{equation}
\label{eq:transfinite}
 \inf_{\pi\in\HR_{\Diamond}} \Exp^{\sigma_{\eps},\pi}_v[S_{|wu|}^{\tau}~|~\run(wu)] \geq L^{d(wu)}(\mathbf{0})_u - \frac{\eps}{2^{|wu|}}.
\end{equation}

In particular, we have
\[
   \inf_{\pi\in\HR_{\Diamond}} \Exp^{\sigma_{\eps},\pi}_v(\Acc) \geq  \inf_{\pi\in\HR_{\Diamond}} \Exp^{\sigma_{\eps},\pi}_v[S_0^{\tau}~|~\run(v)] \geq L^{\alpha}(\mathbf{0})_v - {\eps} = \val_{\HR}(v)-\eps.
\]
\end{lemma}
\begin{proof}
We proceed by transfinite induction on $d(wu)$. If $d(wu)=0$, then the inequality \eqref{eq:transfinite} clearly holds. Now suppose that $d(wu)>0$ and that the inequality \eqref{eq:transfinite} holds for every $\beta<d(wu)$. We distinguish three cases depending on the type of $u$.
\begin{enumerate}
 \item[(1.)] $u \in V_{\Box}$. Denote by $u'$ the successor of $u$ selected by $\sigma_{\eps}(wu)$. Then we have
\begin{align*}
 \inf_{\pi\in\HR_{\Diamond}} \Exp^{\sigma_{\eps},\pi}_v[S_{|wu|}^{\tau}~|~\run(wu)] &= r(u) + \inf_{\pi\in\HR_{\Diamond}} \Exp^{\sigma_{\eps},\pi}_v[S_{|wuu'|}^{\tau}~|~\run(wuu')] \\
  &\geq r(u) + L^{d(wuu')}(\mathbf{0})_{u'} - \frac{\eps}{2^{|wu|+1}} \\&\geq  L^{d(wu)}(\mathbf{0})_u - \frac{\eps}{2^{|wu|}},
\end{align*}
where the second line follows from the induction hypothesis and from the fact that $d(wuu')<d(wu)$, and the third line follows from the definition of $\sigma_{\eps}$.

\item[(2.)] $u \in V_{\Diamond}$. Then we have
\begin{align*}
 \inf_{\pi\in\HR_{\Diamond}} \Exp^{\sigma_{\eps},\pi}_v[S_{|wu|}^{\tau}~|~\run(wu)] &= r(u) + \inf_{u \tran{} u'}\inf_{\pi\in\HR_{\Diamond}} \Exp^{\sigma_{\eps},\pi}_v[S_{|wuu'|}^{\tau}~|~\run(wuu')] \\ &\geq r(u)   + \inf_{u \tran{} u'} L^{d(wuu')}(\mathbf{0})_{u'}- \frac{\eps}{2^{|wu|+1}} \\&\geq L^{d(wu)}(\mathbf{0})_u - \frac{\eps}{2^{|wu|}},
\end{align*}
where the first line is easy, the second line again follows from the induction hypothesis and the third line follows from Lemma \ref{lem:ord-labeling}.

\item[(3.)] $u \in V_{\bigcirc}$. We denote by $u \tran{x} u'$ the fact that $\Prob(u)(u,u')=x$. We have
\begin{align*}
 \inf_{\pi\in\HR_{\Diamond}} \Exp^{\sigma_{\eps},\pi}_v[S_{|wu|}^{\tau}~|~\run(wu)] &= r(u) + \sum_{u \tran{x} u'} x\cdot \left(\inf_{\pi\in\HR_{\Diamond}}\Exp^{\sigma_{\eps},\pi}_v[S_{|wuu'|}^{\tau}~|~\run(wuu')]\right) \\ &\geq r(u)  + \bigg(\sum_{u \tran{x} u'} x\cdot L^{d(wuu')}(\mathbf{0})_{u'}\bigg) - \frac{\eps}{2^{|wu|+1}}\\&\geq L^{d(wu)}(\mathbf{0})_u - \frac{\eps}{2^{|wu|}},
\end{align*}
where again the second and the third line follows from induction hypothesis and Lemma \ref{lem:ord-labeling}, respectively.
\end{enumerate}
\qed
\end{proof}

It remains to show how to handle the case when there are vertices with infinite $\HR$-values. The idea is the same, but the proof is more technical. We need to slightly generalize the previous two lemmas. The following lemma generalizes Lemma \ref{lem:ord-labeling}. We denote by $\last{w}$ the last vertex on a nonempty path $w$.
\begin{lemma}
\label{lem:ord-labeling-inf}
 Under the assumptions of Lemma \ref{lem:ord-labeling} there exists a labeling function $d\colon \fpath(v) \rightarrow \Ord_{\alpha}$ 
satisfying the following conditions:
\begin{enumerate}
 \item[(a)] $d(v)=\alpha$.
 \item[(b)] For every $wu\in\fpath(v)$ it holds either $d(w)=0$ or $d(wu)<d(w)$.
 \item[(c)] For every $wu\in\fpath(v)$, such that $L^{d(wu)}(\mathbf{0})_u<\infty$, we have 
  \[
    L^{d(wu)}(\mathbf{0})_u - \frac{\eps}{2^{|wu|+1}} \leq \begin{cases}
                                                      r(u) + L^{d(wuu')}(\mathbf{0})_{u'}, \text{ for some } u \tran{} u' & \text{ if } u \in V_{\Box}\\
						      r(u) + \inf_{u\tran{} u'} L^{d(wuu')}(\mathbf{0})_{u'} & \text { if } u \in V_{\Diamond}\\
						      r(u) + \sum_{u \tran{} u'}\Prob(u)(u,u')\cdot L^{d(wuu')}(\mathbf{0})_{u'} & \text{ if } u \in V_{\bigcirc},
                                                      \end{cases}
  \]
and for every $wu\in\fpath(v)$, such that $L^{d(wu)}(\mathbf{0})_u=\infty$, we have
\[
    \frac{1}{\eps} + {\eps\cdot(|wu|+1)} + F(w)\leq \begin{cases}
                                                      r(u) + L^{d(wuu')}(\mathbf{0})_{u'}, \text{ for some } u \tran{} u' & \text{ if } u \in V_{\Box}\\
						      r(u) + \inf_{u\tran{} u'} L^{d(wuu')}(\mathbf{0})_{u'} & \text { if } u \in V_{\Diamond}\\
						      r(u) + \sum_{u \tran{} u'}\Prob(u)(u,u')\cdot L^{d(wuu')}(\mathbf{0})_{u'} & \text{ if } u \in V_{\bigcirc},
                                                      \end{cases}
\]
where $F(w)=\begin{cases}
            L^{d(w)}(\vec{0})_{\last{w}} & \text{if } w \text{ is nonempty and } L^{d(w)}(\vec{0})_{\last{w}}<\infty \\
	    0 & \text{otherwise}.
            \end{cases}
$
\end{enumerate}
\end{lemma}
\begin{proof}
 We again define the function $d$ inductively, starting by putting $d(v)=\alpha$. Now let $wu$ be an arbitrary finite path such that $L^{d(wu)}(\vec{0})_u = \infty$. If $d(wu)=\beta+1$ for some ordinal $\beta$, then we can put $d(wuu')=\beta$ for all successors $u'$ of $u$. From the definition of $L$ it then easily follows that the inequality in (c) holds for $wu$. (For example, if $u\in V_{\Box}$, then we have $\infty=r(u)+\sup_{u\tran{}u'
}L^{\beta}(\vec{0})_{u'}$ and there is surely $u\tran{}u'$ s.t. $r(u)+L^{\beta}(\vec{0})_{u'}\geq 1/\eps+\eps\cdot(|wu|+1)+F(w)$. It is of course possible that $L^{\beta}(\vec{0})_{u'}=\infty$.)

If $d(wu)$ is an limit ordinal, then there is a successor ordinal $\beta+1<d(wu)$ s.t. $L^{\beta+1}(\vec{0})_{u} \geq 2/\eps + \eps\cdot(|wu|+1)+F(w)$. We set $d(wuu')=\beta$ for all successors $u'$ of $u$. If $L^{\beta+1}(\vec{0})_{u}=\infty$, then from the previous paragraph we get that (c) holds for $wu$. If $L^{\beta+1}(\vec{0})_{u} <\infty$, then the same argument as in the proof of Lemma~\ref{lem:ord-labeling} shows, that for every $\delta>0$ the right-hand side of the inequality in (c) is $\delta$-close to $L^{\beta+1}(\vec{0})_0$. If we set $\delta=1/\eps$, we get that (c) holds for $wu$.

For $wu$ with $L^{d(wu)}(\vec{0})_u<\infty$ we can use the same construction as in the Lemma~\ref{lem:ord-labeling}.
\qed
\end{proof}

For every $wu$ let us set
\[
 A_{\eps}^{wu} = \begin{cases}
                  L^{d(wu)}(\mathbf{0})_u - \frac{\eps}{2^{|wu|+1}} & \text{ if } L^{d(wu)}(\mathbf{0})_u < \infty\\
		   \frac{1}{\eps} + {\eps\cdot(|wu|+1)} + F(w)& \text{otherwise},
                 \end{cases}
\]

and 
\[
 B_{\eps}^{wu} = \begin{cases}
                  L^{d(wu)}(\mathbf{0})_u - \frac{\eps}{2^{|wu|}} & \text{ if } L^{d(wu)}(\mathbf{0})_u < \infty\\
		   \frac{1}{\eps} + {\eps\cdot|wu|} + F(w)& \text{otherwise}.
                 \end{cases}
\]
Note that $A_{\eps}^{wu}-\delta\geq B_{\eps}^{wu}$ for every $0\leq \delta\leq \eps/2^{|wu|+1}$.
We now define the $\eps$-$\HR$-optimal deterministic strategy $\sigma_{\eps}$ as follows: for a given $wu\in\fpath(v)$, the $\sigma(wu)$ selects a transition $u\tran{}u'$ such that $A_{\eps}^{wu} \leq r(u) + L^{d(wuu')}(\mathbf{0})_{u'}$. It remains to prove that $\sigma_{\eps}$ is $\eps$-$\HR$-optimal in $v$. We generalize Lemma \ref{lem:eps-HD-strat} as follows:

\begin{lemma}
 The following holds for every $wu\in\fpath(v)$:
\begin{equation}
 \label{eq:transfinite-inf}
 \inf_{\pi\in\HR_{\Diamond}} \Exp^{\sigma_{\eps},\pi}_v[S_{|wu|}^{\tau}~|~\run(wu)] \geq B_{\eps}^{wu}.
\end{equation}
\end{lemma}
\begin{proof}

 The proof again proceeds by transfinite induction on $d(wu)$. The base case is the same as in Lemma \ref{lem:eps-HD-strat}, because if $d(wu)=0$, then $B_{\eps}^{wu}=-\frac{\eps}{2^{|wu|+1}}$. So assume that $d(wu)>0$ and that \eqref{eq:transfinite-inf} hols for all $\alpha<d(wu)$. If $L^{d(wu)}(\mathbf{0})_u < \infty$, then we can basically proceed in exactly the same way as in the Lemma \ref{lem:eps-HD-strat}. The only difference here is the case when $u \in V_{\Diamond}$, $L^{d(wu)}(\mathbf{0})_u < \infty$ and $L^{d(wuu')}(\mathbf{0})_{u'} = \infty$ for some $u \tran{} u'$. But in this case we have $\Exp^{\sigma_{\eps},\pi}_v[S_{|wuu'|}^{\tau}~|~\run(wuu')]\geq  B^{wuu'}_{\eps} > 1/\eps + F(wu) = 1/\eps + L^{d(wu)}(\vec{0})_u \geq 1/\eps + \inf_{u \tran{} u'}L^{d(wuu')}(\vec{0})_{u'}$, so the computation in part (2.) of the proof of Lemma \ref{lem:eps-HD-strat} is still valid.


 If $L^{d(wu)}(\mathbf{0})_u = \infty$, then we consider the following cases:

\begin{enumerate}
 \item[(1.)] $u \in V_{\Box}$. Denote by $u'$ the successor of $u$ selected by $\sigma_{\eps}(wu)$. Then 
\begin{align*}
 \inf_{\pi\in\HR_{\Diamond}} \Exp^{\sigma_{\eps},\pi}_v[S_{|wu|}^{\tau}~|~\run(wu)] &= r(u) + \inf_{\pi\in\HR_{\Diamond}} \Exp^{\sigma_{\eps},\pi}_v[S_{|wuu'|}^{\tau}~|~\run(wuu')] \\
  &\geq r(u) + B_{\eps}^{wuu'},
\end{align*}
where the second line comes from the induction hypothesis. There are two possibilities. Either \begin{equation}\label{eq:hd1}B_{\eps}^{wuu'} = 1/\eps + \eps\cdot{|wu|}+\eps+ F(w)> 1/\eps + \eps\cdot{|wu|} +F(w)=B_{\eps}^{wu},\end{equation} or \begin{equation}\label{eq:hd2} r(u) + B_{\eps}^{wuu'} = r(u) +  L^{d(wuu')}(\mathbf{0})_{u'} - \frac{\eps}{2^{|wu|+1}} \geq A_{\eps}^{wu}- \frac{\eps}{2^{|wu|+1}} \geq B_{\eps}^{wu},\end{equation} where the second inequality follows from Lemma \ref{lem:ord-labeling-inf} and from the definition of $\sigma_{\eps}$. In both cases the equation \eqref{eq:transfinite-inf} holds.

\item[(2.)] $u \in V_{\Diamond}$. Then we have
\begin{align*}
 \inf_{\pi\in\HR_{\Diamond}} \Exp^{\sigma_{\eps},\pi}_v[S_{|wu|}^{\tau}~|~\run(wu)] &= r(u) + \inf_{u \tran{} u'}\inf_{\pi\in\HR_{\Diamond}} \Exp^{\sigma_{\eps},\pi}_v[S_{|wuu'|}^{\tau}~|~\run(wuu')] \\ &\geq \inf_{u \tran{} u'}\left(r(u) +  B_{\eps}^{wuu'}\right).
\end{align*}
Exactly the same computation as in the case (1.) reveals that \eqref{eq:hd1} or \eqref{eq:hd2} holds for all $u \tran{} u'$, and thus for all these transitions we have $r(u) + B_{\eps}^{wuu'}\geq B_{\eps}^{wu}$. Thus, $\inf_{u \tran{} u'}\left(r(u) +  B_{\eps}^{wuu'}\right)\geq B_{\eps}^{wu}$ and \eqref{eq:transfinite-inf} holds for $wu$.

\item[(3.)] $u \in V_{\bigcirc}$. Then again from the induction hypothesis it follows that
\begin{align*}
 \inf_{\pi\in\HR_{\Diamond}} \Exp^{\sigma_{\eps},\pi}_v[S_{|wu|}^{\tau}~|~\run(wu)] &= r(u) + \sum_{u \tran{x} u'} x\cdot \left(\inf_{\pi\in\HR_{\Diamond}}\Exp^{\sigma_{\eps},\pi}_v[S_{|wuu'|}^{\tau}~|~\run(wuu')]\right) \\ &\geq \sum_{u \tran{x} {u'}}x\cdot\left(r(u) + B_{\eps}^{wuu'}\right )\geq B_{\eps}^{wu},
\end{align*}
where the last inequality can be justified in exactly the same way as in the previous two cases.
\end{enumerate}
\qed
\end{proof}

\section{Proof of Lemma~\ref{lem-max-bound}}
\label{app:nstep}
\begin{reflemma}{lem-max-bound} 
If $G$ is $\Diamond$-finitely-branching, then for every 
$v \in V$ there is $n\in \Nset$ such that
\begin{equation}
  \label{eq:max-bound}
  \adjustlimits\sup_{\sigma\in\HR_{\Box}} \inf_{\pi\in\HR_{\Diamond}} \Exp^{\sigma,\pi}_v(\Acc_n) 
  \quad > \quad \val_{\HR}(v) \ominus \frac{\varepsilon}{4}
\end{equation}
\end{reflemma}

Let $v\in V$ be arbitrary.
Without loss of generality, we can assume that $v\in V_{\bigcirc}$ and that $v$ has only one outgoing transition. If this is not the case, we can simply add a new stochastic vertex $v'$ with a zero reward and a single new transition $v\tran{} v'$. It is clear, that if the statement of the lemma holds for $v'$ in this new game, then it holds for $v$ in the original game.


Observe that if every vertex of player $\Diamond$ has only finitely many successors, then the operator $L$ is Scott-continuous.

\begin{lemma}
Let $D \subseteq (\Rset^{\geq0}_{\infty})^V$ be an arbitrary directed set (i.e. such a set that each pair of elements in $D$ has an upper bound in $D$.) Then $L(\sup_{\vec{d}\in D}\vec{d}) = \sup_{\vec{d} \in D}L(\vec{d})$.
\end{lemma}
\begin{proof}
 The inequality $\geq$ follows immediately from the monotonicity of $L$. So it suffices to prove that for every directed set $D$ and every vertex $v$ we have $L(\sup_{\vec{d}\in D}\vec{d})_v \leq \sup_{\vec{d} \in D}L(\vec{d})_v$. 
Note that $(\sup_{\vec{d} \in D}\vec{d})_v = \sup_{\vec{d}\in D}\vec{d}_v$. We consider three cases:
\begin{itemize}
 \item[(1.)] $v \in V_{\Box}$. Then we trivially have
\begin{align*}
 L(\sup_{\vec{d}\in D}\vec{d})_v=\adjustlimits\sup_{v \tran{} v'} \sup_{\vec{d} \in D}\vec{d}_{v'} =  \adjustlimits\sup_{\vec{d} \in D}\sup_{v \tran{} v'}\vec{d}_{v'} = \sup_{\vec{d} \in D}L(\vec{d})_v.
\end{align*}
\item[(2.)] $v \in V_{\Diamond}$. Assume, for the sake of contradiction, that $\inf_{v\tran{}v'}\sup_{\vec{d}\in D}\vec{d}_{v'} > \sup_{\vec{d} \in D} \inf_{v\tran{}v'}\vec{d}_{v'}$. Then for each of the finitely many transitions $v\tran{}v'$ there is a vector $\vec{d}(v')\in D$ such that $\vec{d}(v')_{v'}>\sup_{\vec{d} \in D} \inf_{v\tran{}v'}\vec{d}_{v'}$. But since the set $D$ is directed and there are only finitely many $v \tran{} v'$, there is a vector  $\vec{d}^* \in D$ such that $\vec{d}(v')\vleq \vec{d}^*$ for every successor $v'$ of $v$. We thus have
\begin{align*}
 \adjustlimits\sup_{\vec{d}\in D}\inf_{v\tran{}{v'}} \vec{d}_{v'} \geq \inf_{v\tran{}v'} \vec{d}^*_{v'} \geq \inf_{v\tran{}v'} \vec{d}(v')_{v'} >  \adjustlimits\inf_{v\tran{}v'}\sup_{\vec{d} \in D} \inf_{v\tran{}v'}\vec{d}_{v'} = \adjustlimits\sup_{\vec{d} \in D} \inf_{v\tran{}v'}\vec{d}_{v'},
\end{align*}
a contradiction. (Above, the second inequality follows from the fact that $\vec{d}(v')\vleq \vec{d}^*$ for every $v'$ and the first inequality and the last equality are trivial. The third inequality is strict because there are only finitely many successors of $v$.)
\item[(3.)] $v\in V_{\bigcirc}$. Then we again trivially have 
\[
 L(\sup_{\vec{d}\in D}\vec{d})_v=\sum_{v \tran{} v'}\Prob(v)(v,v')\cdot\sup_{\vec{d} \in D}\vec{d}_{v'} = \sup_{\vec{d}\in D}\sum_{v \tran{} v'}\Prob(v)(v,v')\cdot\vec{d}_{v'} = \sup_{\vec{d} \in D}L(\vec{d})_v.
\]
\end{itemize}
\qed
\end{proof}

From the Kleene fixed-point theorem it follows that $L^{\omega}(\vec{0})=\mathbf{K}$, i.e. that the ordinal number $\alpha$ from Lemmas \ref{lem:ord-labeling} and \ref{lem:ord-labeling-inf} can be assumed to be equal to $\omega$.
Fix a labeling $d$ of finite paths starting in $v$ that satisfies the conditions (a)--(c) in Lemma \ref{lem:ord-labeling} (or Lemma \ref{lem:ord-labeling-inf}, if there are some vertices with infinite $\HR$-value). Then $v$ is labeled by $\omega$ and all other elements of $\fpath(v)$ are labeled with nonnegative integers. Recall that $\tau(\omega)$ denotes the least $k$ such that $d(\omega(0),\dots,\omega(k))=0$.

Now let $u$ be the unique successor of $v$. We set $n=d(vu)+1$. To see that this $n$ satisfies \eqref{eq:max-bound}, consider the deterministic $(\eps/8)$-$\HR$-optimal strategy $\sigma_{\eps/8}$ constructed in the proof of Lemma \ref{lem:eps-HD-opt}. From Lemma \ref{lem:ord-labeling} (or Lemma \ref{lem:ord-labeling-inf}) it follows that 
\begin{equation*}\inf_{\pi\in\HR_{\Diamond}} \Exp^{\sigma_{\eps/8},\pi}_v[\sum_{i=0}^{\tau(\omega)}r(\omega(i))~|~\run(v)] \geq \val_{\HR} \ominus \frac{\eps}{8}.\end{equation*} 

But now we clearly have $\tau(\omega)\leq n=d(vu)+1$ for all runs $\omega$ starting in $v$. Thus, we have
\begin{equation*}
 \inf_{\pi\in\HR_{\Diamond}} \Exp^{\sigma_{\eps/8},\pi}_v(\Acc_n)\geq\inf_{\pi\in\HR_{\Diamond}} \Exp^{\sigma_{\eps/8},\pi}_v[\sum_{i=0}^{\tau(\omega)}r(\omega(i))~|~\run(v)]\geq \val_{\HR} \ominus \frac{\eps}{8} > \val_{\HR} \ominus \frac{\eps}{4}.
\end{equation*}
This finishes the proof of Lemma \ref{lem-max-bound}.

\section{MD-optimal strategies for player $\Diamond$}
\label{app:MD-opt-min}

We prove Item~\ref{prop:finb-MD-diam} of 
Proposition~\ref{prop:finb-MD}, i.e. the fact that for every $\Diamond$-finitely-branching game $G$ there is $\pi \in \MD_{\Diamond}$ such that $\pi$ is $\HR$-optimal in every vertex.
We have already defined $\pi$ as follows: In every state $v\in V_{\Diamond}$, the 
strategy $\pi$ chooses a successor $u$ minimizing $\val_{\HR}(u)$ 
among all successors of $v$. But the $\HR$-optimality of this strategy immediately follows from Lemma \ref{lem:eps-min} (note that this lemma works for $\eps=0$) and Lemma \ref{lem:hrdet} (which says that the least fixed-point $\vec{K}$ of $L$ is equal to the vector of $\HR$-values).


%
\end{document}